\documentclass[11pt]{article}

\usepackage[margin=1in]{geometry}
\usepackage{amsmath,amssymb,amsthm,mathtools}
\usepackage{enumitem}
\usepackage{microtype}
\usepackage{hyperref}
\hypersetup{hidelinks}

\newtheorem{theorem}{Theorem}[section]
\newtheorem{lemma}[theorem]{Lemma}
\newtheorem{proposition}[theorem]{Proposition}
\newtheorem{corollary}[theorem]{Corollary}

\theoremstyle{definition}

\theoremstyle{remark}
\newtheorem{remark}[theorem]{Remark}

\newcommand{\SC}{\operatorname{SC}}
\newcommand{\plu}{\operatorname{plu}}
\newcommand{\IV}{\operatorname{IV}}
\newcommand{\ML}{\operatorname{ML}}
\newcommand{\MIV}{\operatorname{MIV}}
\newcommand{\dist}{\operatorname{dist}}
\newcommand{\supp}{\operatorname{supp}}

\newcommand{\eps}{\varepsilon}

\title{Robust Lottery Compression for Metric Voting:\\
A Transfer Principle for Bounded Randomness}
\author{
Jianhao Jia\thanks{Shanghai University of Finance and Economics.
Email: \texttt{jianhao.jia@stu.sufe.edu.cn}}
\and
Bo Peng\thanks{Shanghai University of Finance and Economics.
Email: \texttt{ahqspbo@gmail.com}}
}
\date{}

\begin{document}
\maketitle

\begin{abstract}
We study metric distortion in randomized social choice under \emph{bounded randomness}: on every preference profile, the voting rule must deterministically identify a multiset of \(K\) candidates and then select a uniformly random entry. Previous work showed that this restricted model can beat the optimal deterministic distortion of \(3\). We show that it can in fact approach the current best unrestricted upper benchmark of \(5/2\). For every integer \(K\ge 802\), there exists a bounded-randomness rule with distortion at most
\[
\frac{5}{2}
+3\left(\frac{\pi}{8K}\right)^{1/3}
+2\sqrt{\frac{\pi}{8K}}.
\]
Consequently, \(O(\varepsilon^{-3})\) entries suffice for distortion \(5/2+\varepsilon\), independently of the numbers of voters and candidates. We also show that \(164\) entries already achieve distortion strictly below \(3\), giving \(2\le N^\star\le 164\) for the minimum list size needed to break the deterministic barrier. Our main technical contribution is a dimension-free compression theorem: if a lottery has distortion at most \(\rho\) and every candidate in its support has deterministic distortion at most \(H\), then it admits a uniform \(K\)-entry approximation with distortion at most \(\rho+(H+1)\sqrt{\pi/(8K)}\). Thus, lotteries whose possible outcomes are already well behaved incur only \(O(K^{-1/2})\) compression loss. Mixed Integrated Veto does not satisfy this support condition, so we first remove early-eliminated outcomes, trading $O(\tau^2)$ distortion loss for an $O(1/\tau)$ bound on the
deterministic distortion of every supported candidate. Balancing this repair cost against compression yields the \(O(K^{-1/3})\) convergence rate.
\end{abstract}

\section{Introduction}
\label{sec:intro}
Randomization is one of the few general ways in which a voting rule can improve upon the welfare guarantees of deterministic social choice. In the standard model, however, a randomized rule may output an arbitrary probability distribution over the candidates. Such a lottery may have large support and may use finely tuned, nonuniform probabilities. This mathematical freedom is useful, but it leaves open how much probabilistic complexity is actually needed to obtain the benefits of randomization.

We study a particularly simple alternative. On every preference profile, the voting rule deterministically identifies a multiset of \(K\) candidates, where repetitions are allowed, and then selects a uniformly random entry. Repetitions encode unequal probabilities through multiplicity, but the entire lottery is represented by an explicit \(K\)-entry list. Most importantly, \(K\) is independent of the numbers of voters and candidates. Following Cai, Gao, Ramakrishnan, and Wang~\cite{CGRW26}, we call this model \emph{bounded randomness}.

The restriction is severe. A general randomized rule can assign an arbitrary probability to every candidate, whereas a bounded-randomness rule must expose all of its possible outcomes in a short list. This makes the resulting lottery transparent and easy to describe, but it is not clear whether such simplicity is compatible with the strongest welfare guarantees known for randomized voting.

We investigate this question in the metric distortion framework. Voters and candidates are embedded in an unknown metric space, and each voter ranks the candidates by increasing distance. A voting rule observes only the ordinal preference profile and seeks to minimize social cost. The distortion of a randomized voting rule is the worst-case ratio, over all preference profiles and all consistent metrics, between the expected social cost of its output and the minimum social cost of any candidate.

The optimal deterministic distortion is \(3\)~\cite{ABP18,GHS20}. Randomization can do strictly better. Charikar and Ramakrishnan established the current best lower bound of \(2.1126\) for unrestricted randomized rules~\cite{CR22}. On the upper-bound side, Charikar, Ramakrishnan, Wang, and Wu were the first to obtain a constant improvement over \(3\), achieving distortion \(2.75271\)~\cite{CRWW24}. More recently, Frank and Ye independently showed that an equal mixture of a Maximal Lottery and a lottery obtained by averaging candidate scores over the Simultaneous Veto process has distortion \(5/2\)~\cite{Frank26,Ye26}. Frank calls the second component \emph{Integrated Veto} and the resulting rule \emph{Mixed Integrated Veto}; we adopt this terminology.

Cai et al.~\cite{CGRW26} initiated the study of metric distortion under bounded randomness. They proved that a uniform lottery over a constant-size, deterministically identified multiset already achieves distortion strictly below \(3\). Their result shows that arbitrary probability vectors are not necessary to separate randomized voting from deterministic voting. They also ask for the smallest list size \(N\) that suffices to beat \(3\), and conjecture that \(N=2\) may already be enough.

Their work raises a broader question. Is bounded randomness useful only for crossing the deterministic barrier, or can a short uniform list preserve the substantially stronger guarantees of unrestricted randomized voting?

\begin{quote}
\emph{Can bounded randomness approach the current best unrestricted distortion guarantee? More quantitatively, how does the achievable distortion depend on the list size \(K\)?}
\end{quote}

We give a quantitative affirmative answer. Bounded randomness approaches the current unrestricted upper benchmark \(5/2\), with a list size depending only on the desired accuracy. We also obtain an explicit answer to the finite-list question: \(164\) entries already suffice to achieve distortion strictly below \(3\). The next subsection states these results and explains the general compression principle behind them.

\subsection{Our Results}

Our contribution has three parts. First, we give an explicit
list-size--distortion tradeoff approaching \(5/2\). Second, we prove a
general theorem identifying when an unrestricted lottery can be replaced
by a short uniform list. Third, a sharper finite-sample analysis yields
the concrete bound \(N^\star\le 164\).

Our main quantitative result is the following.

\begin{theorem}\label{thm:main-K}
For every integer \(K\ge 802\), there exists a voting rule that
deterministically identifies a multiset of \(K\) candidates and uniformly
randomizes over it, with metric distortion at most
\[
\frac52
+3\left(\frac{\pi}{8K}\right)^{1/3}
+2\sqrt{\frac{\pi}{8K}}.
\]
\end{theorem}

In particular, the distortion converges to \(5/2\) as the list size grows.
Equivalently, a list of size \(O(\eps^{-3})\), independent of the numbers
of voters and candidates, suffices to obtain distortion \(5/2+\eps\).

\begin{theorem}\label{thm:main-eps}
For every constant \(\eps>0\), there exists an integer \(K=K(\eps)\),
independent of the numbers of voters and candidates, and a voting rule
such that
\begin{enumerate}[label=(\roman*)]
\item on every preference profile, the rule deterministically identifies
a multiset of \(K\) candidates and then selects a uniformly random entry
of the multiset;
\item the metric distortion of the rule is at most \(\frac52+\eps\).
\end{enumerate}
Moreover, for
\(\bar\eps=\min\left\{\eps,\frac15\right\}\), it suffices to take
$K\ge
\left\lceil
\frac{27\pi}{8\bar\eps^3}
\left(1+\sqrt{\frac{\bar\eps}{3}}\right)^2
\right\rceil$.
For each fixed \(\eps\), the multiset can be found in polynomial time.
\end{theorem}

The polynomial-time guarantee is for every fixed accuracy parameter; the
degree of the polynomial may depend on \(K(\eps)\).

\paragraph{Uniform-support compression.}

Low expected distortion alone does not guarantee that a lottery can be
compressed. A source lottery may assign very small probability to a
candidate of extremely large deterministic distortion, whereas every
candidate appearing in a \(K\)-entry uniform list receives probability
at least \(1/K\).

Suppose a lottery \(D\) has distortion at most \(\rho\), and every
candidate in \(\supp(D)\) has deterministic distortion at most \(H\),
uniformly over all metrics consistent with the profile. We prove that,
for every \(K\ge1\), there exists a uniform lottery over a \(K\)-element
multiset with distortion at most
\begin{equation}\label{eq:intro-compression}
\rho+(H+1)\sqrt{\frac{\pi}{8K}}.
\end{equation}

This is the Uniform-Support Compression Theorem proved in
Section~\ref{sec:compression}. Its bound is independent of the numbers of
voters and candidates. Thus, when \(H\) is an absolute constant, the
compression loss is \(O(K^{-1/2})\). We also prove an exact finite-sample
version, which is used in Theorem~\ref{thm:164}.

\paragraph{Direct compression of the \(2.75271\) benchmark.}

Cai et al.~\cite{CGRW26} observe that every candidate supported by the
\(2.75271\)-distortion rule of Charikar, Ramakrishnan, Wang, and
Wu~\cite{CRWW24} has deterministic distortion bounded by an absolute
constant. Therefore, for some absolute constant
\(C_{\mathrm{CRWW}}>0\), our compression theorem gives
\begin{equation}\label{eq:crww-direct-intro}
\dist
\le
2.75271+\frac{C_{\mathrm{CRWW}}}{\sqrt K}.
\end{equation}
Hence \(O(\eps^{-2})\) entries suffice to preserve this benchmark up to
additive error \(\eps\).

\paragraph{Robustifying Mixed Integrated Veto.}

Mixed Integrated Veto achieves the sharper benchmark \(5/2\), but it
cannot be compressed directly because its Integrated-Veto component may
support candidates of arbitrarily large deterministic distortion.

We therefore remove the Integrated-Veto probability assigned to
candidates eliminated before a threshold time \(\tau\). The removed mass
is at most \(\tau^2\), while every surviving candidate has deterministic
distortion at most \(1+2/\tau\). Combining the surviving mass with a
Maximal Lottery produces a source lottery whose distortion is at most
$1+\frac{3}{2-\tau^2}$
and whose supported candidates have deterministic distortion
\(O(1/\tau)\).

Applying the compression theorem therefore gives excess distortion
$O(\tau^2)+O\left(\frac{1}{\tau\sqrt K}\right)$.
Choosing \(\tau=\Theta(K^{-1/6})\) yields $\dist=\frac52+O(K^{-1/3})$.
Tracking the constants with
\(\tau=(\pi/(8K))^{1/6}\) gives Theorem~\ref{thm:main-K}.

This also explains the two convergence rates:
\[
2.75271+O(K^{-1/2})
\qquad\text{and}\qquad
\frac52+O(K^{-1/3}).
\]
The first is the cost of directly compressing an already robust lottery.
The second also includes the cost of robustifying Mixed Integrated Veto.

\paragraph{A concrete list breaks the deterministic barrier.}

Let \(N^\star\) be the minimum integer \(N\) for which some voting rule
uniformly randomizing over a deterministically identified \(N\)-element
multiset has distortion strictly smaller than \(3\).

\begin{theorem}\label{thm:164}
There exists a voting rule that deterministically identifies a multiset
of \(164\) candidates and uniformly randomizes over it, with metric
distortion strictly smaller than \(3\). Consequently,
\(2\le N^\star\le164\).
\end{theorem}

The lower bound \(N^\star\ge2\) is the optimal deterministic lower
bound~\cite{ABP18}. Thus, Theorem~\ref{thm:164} does not resolve the
conjecture \(N^\star=2\), but it replaces the previously qualitative
constant-size upper bound by an explicit value.

The proof uses the same robustified source lottery and compression
argument, with Smirnov's exact finite-sample formula~\cite{Smirnov44}
providing the sharper estimate needed for the explicit \(164\)-entry
bound.

\paragraph{List size and limiting distortion.}

Let \(\rho_K\) denote the infimum distortion among voting rules that
uniformly randomize over a deterministically identified \(K\)-element
multiset. The two constructions above imply
\begin{equation}\label{eq:rhoK-upper-intro}
\rho_K
\le
\min\left\{
2.75271+\frac{C_{\mathrm{CRWW}}}{\sqrt K},
\;
\frac52
+3\left(\frac{\pi}{8K}\right)^{1/3}
+2\sqrt{\frac{\pi}{8K}}
\right\},
\qquad K\ge802.
\end{equation}

Define
 $\rho_{\mathrm{BR}}:=\inf_{K\ge1}\rho_K$,
where \(K\) may depend on the desired accuracy but not on the numbers of
voters and candidates. Theorem~\ref{thm:main-eps} implies
\begin{equation}\label{eq:rhoBR}
\rho_{\mathrm{BR}}\le\frac52.
\end{equation}

Thus bounded randomness approaches the current best unrestricted upper
benchmark. This does not prove that bounded and unrestricted
randomization have the same optimal distortion: the unrestricted optimum
may be below \(5/2\), and our compression theorem requires control of the
individual candidates supported by the source lottery.

\paragraph{Committee selection with repeated seats.}

A uniform \(K\)-entry list can also be interpreted as a deterministic
committee with \(K\) seats when repetitions are allowed and committee
cost is additive over the seats~\cite{CGRW26,HA25}.

\begin{corollary}\label{cor:committee}
Consider the committee-selection setting in which a committee is a
multiset of \(K\) candidates, repetitions are allowed, and its social cost
is the sum of the social costs of its \(K\) entries. Then:
\begin{enumerate}[label=(\roman*)]
\item for every \(K\ge802\), there is a deterministic
committee-selection rule with distortion at most
\[
\frac52
+3\left(\frac{\pi}{8K}\right)^{1/3}
+2\sqrt{\frac{\pi}{8K}};
\]
\item a committee of \(164\) seats already suffices for distortion
strictly smaller than \(3\).
\end{enumerate}
\end{corollary}

\begin{proof}
For a multiset \(S\) of size \(K\), its committee social cost is \(K\)
times the expected social cost of a uniformly random entry of \(S\).
Because repetitions are allowed, an optimal \(K\)-seat committee consists
of \(K\) copies of a minimum-social-cost candidate and therefore has cost
\(K\min_j\SC(j)\). Hence the distortion of \(S\) is exactly the distortion
of the uniform lottery over \(S\). The claims follow from
Theorems~\ref{thm:main-K} and~\ref{thm:164}.
\end{proof}

This consequence relies on allowing repeated winners. It does not extend
directly to models requiring distinct committee members. Finally,
Theorem~\ref{thm:main-eps} gives a finite list size for every fixed
\(\eps>0\), but does not show that one fixed finite list size attains
distortion exactly \(5/2\).

\subsection{Further Related Work}

\paragraph{Metric distortion.}

Procaccia and Rosenschein introduced distortion as a measure of the welfare loss from using ordinal rather than cardinal preferences~\cite{DBLP:conf/cia/ProcacciaR06}. Anshelevich, Bhardwaj, and Postl and Anshelevich et al. developed the metric model and analyzed classical voting rules~\cite{DBLP:conf/aaai/AnshelevichBP15,ABP18}; see also the survey of Anshelevich et al.~\cite{DBLP:conf/ijcai/AnshelevichF0V21}. For deterministic voting, work by Munagala and Wang and Kempe~\cite{DBLP:conf/ec/MunagalaW19,DBLP:conf/aaai/Kempe20a} culminated in the optimal distortion-\(3\) rule of Gkatzelis, Halpern, and Shah~\cite{GHS20}, followed by simpler optimal rules of Kizilkaya and Kempe~\cite{DBLP:conf/ijcai/KizilkayaK22,KK23}.

For randomized voting, Random Dictatorship achieves distortion \(3\)~\cite{DBLP:conf/sigecom/FeldmanFG16,DBLP:journals/jair/AnshelevichP17}, while Goel, Krishnaswamy, and Munagala initiated the study of randomized lower bounds and fairness guarantees~\cite{DBLP:conf/sigecom/GoelKM17}. Charikar and Ramakrishnan and, independently, Pulyassary and Swamy ruled out distortion \(2\)~\cite{CR22,pulyassary2021randomized}. Charikar, Ramakrishnan, Wang, and Wu obtained the first constant improvement over \(3\), achieving distortion \(2.75271\)~\cite{CRWW24}. Frank and Ye independently improved the upper bound to \(5/2\) using an equal mixture of a Maximal Lottery and an integrated veto-score lottery~\cite{Frank26,Ye26}; Ye additionally proves optimality within a broader family allowing profile-dependent mixture weights and adaptive weighting over the veto process.

\paragraph{Bounded randomness, transparency, and restricted information.}

Cai et al.~\cite{CGRW26} introduced the bounded-randomness model and showed that a uniform lottery over a constant-size, deterministically identified multiset can beat distortion \(3\). We study how closely such short uniform lotteries can preserve unrestricted distortion guarantees.

Related work considers simple or transparent randomization in other social-choice settings. Flanigan, Kehne, and Procaccia study transparent randomization over fair panels in sortition~\cite{FKP21}, while Ebadian, Filos-Ratsikas, Latifian, and Shah study explainable randomized voting rules in a normalized-utility model, including rules that select uniformly from small committees~\cite{EFLS23}. Other work restricts the information available to a voting rule, through sample-complexity bounds~\cite{DBLP:conf/aaai/GrossAX17,DBLP:conf/aaai/FainGMP19} or communication--distortion tradeoffs~\cite{DBLP:conf/aaai/Kempe20b,DBLP:conf/nips/MandalPSW19,DBLP:conf/sigecom/MandalSW20}. In contrast, bounded randomness permits access to the full ordinal profile and restricts only the complexity of the output lottery.

\paragraph{Lottery constructions and compression.}

Maximal Lotteries are classical probabilistic social-choice rules defined through equilibrium in the Condorcet game~\cite{kreweras1965aggregation,fishburn1984probabilistic,brandt2017rolling}. They play a central role in the randomized-distortion rule of Charikar et al.~\cite{CRWW24}, the bounded-randomness construction of Cai et al.~\cite{CGRW26}, and the \(5/2\)-distortion rules of Frank and Ye~\cite{Frank26,Ye26}. The veto component of the latter rules is based on the Simultaneous Veto process of Kizilkaya and Kempe~\cite{KK23}; Frank calls the resulting average-score rule Integrated Veto, while Ye calls it the Veto Lottery. Rather than proposing a new unrestricted rule, we study how to robustify and compress these lotteries under bounded randomness.

Charikar, Ramakrishnan, and Wang show that sampling from a Maximal Lottery yields constant-support approximate dominating distributions and approximate Maximal Lotteries~\cite{CRW26}, and Cai et al.~\cite{CGRW26} develop related sampling constructions for representative approximate lotteries. In contrast, our compression theorem preserves metric distortion directly, under a realization-wise robustness condition, rather than requiring the empirical lottery to retain an approximate-equilibrium property.

\paragraph{Committee selection and undominated sets.}

Metric distortion has also been studied in multiwinner voting and representative selection~\cite{DBLP:conf/sigecom/ChengDK17,DBLP:conf/aaai/ChengDK18,DBLP:journals/ai/CaragiannisSV22}. Han and Anshelevich study the compatibility of sum and max objectives for committee selection and \(k\)-facility location~\cite{HA25}. Our committee interpretation concerns the repeated-seat model, in which a committee is a multiset and its cost is additive over its entries.

Charikar et al. study small committees with strong majority guarantees~\cite{CLR25}, while Peters relates the support of low-distortion randomized rules to undominated and Condorcet-winning sets~\cite{Pet26}. Together, these results imply the complementary lower bound \(\rho_K\ge 1+\frac{4}{K+1}\), and hence \(\rho_2\ge 7/3\).

\paragraph{Empirical-process bounds.}

Our compression analysis uses classical bounds for empirical distribution functions. Dvoretzky, Kiefer, and Wolfowitz established the foundational distribution-free concentration inequality~\cite{DKW56}, and Massart proved the sharp constant in its two-sided form~\cite{Massart90}. We use the unrestricted one-sided form proved by Reeve~\cite{Reeve24}, together with Smirnov's exact finite-sample distribution of the one-sided Kolmogorov--Smirnov statistic~\cite{Smirnov44}.

\section{Preliminaries}\label{sec:prelim}
Let $V$ be a finite set of voters and $C$ a finite set of candidates.
Each voter $v\in V$ has a strict total order $\succ_v$ over $C$.
We write $\succeq_v$ for the weak order induced by $\succ_v$; that is,
$i\succeq_v j$ if either $i\succ_v j$ or $i=j$.
A preference profile is denoted by $P=(\succ_v)_{v\in V}$.
A lottery is a probability distribution $p\in\Delta(C)$.

A pseudometric $d$ on $V\cup C$ is consistent with $P$ if
$i\succ_v j\implies d(v,i)\le d(v,j)$ for every voter $v$ and
candidates $i,j$. We write $d\models P$ when $d$ is consistent with
$P$. The social cost of a candidate $j$ is
$\SC(j)=\frac1{|V|}\sum_{v\in V}d(v,j)$. For a lottery $p$, write
$\SC(p)=\mathbb{E}_{j\sim p}[\SC(j)]$. Its distortion on a profile is
\[
\dist(p,P)=\sup_{d\models P}
\frac{\SC(p)}{\min_{i\in C}\SC(i)}.
\]
with the usual convention that the ratio is finite only when the numerator is also zero whenever the denominator is zero. The metric distortion of a voting rule is the supremum over all finite profiles.

\subsection{Set-based notation and Frank's certificate}
For candidates $i,j$, let $s_{i\succ j}=\Pr_{v\sim V}[i\succ_v j]$.\footnote{Here and below,
$v\sim V$ means that $v$ is drawn uniformly at random from the finite
voter set $V$; hence these probabilities are simply fractions of
voters.}
For $I\subseteq C$ and $j\in C$, let $s_{I\succ j}=\Pr_{v\sim V}[i\succ_v j\text{ for every }i\in I]$.
If $j\in I$, this quantity is $0$. For $i\in I$, let
$s_{i\succ I^c}=\Pr_{v\sim V}[i\succ_v j\text{ for every }j\in I^c]$.
For every nonempty proper subset $I\subsetneq C$, define
\begin{equation}
q(I)=\min_{i\in I}\bigl(1-s_{i\succ I^c}\bigr).
\end{equation}
For every nonempty subset $I\subseteq C$ and every lottery $p$, define
\begin{equation}
\ell_I(p)=\sum_{j\in I^c}s_{I\succ j}p(j).
\end{equation}
In particular, $\ell_C(p)=0$ for every lottery $p$.
Let $\plu(j)$ be the fraction of voters ranking $j$ first and $\plu(S)=\sum_{j\in S}\plu(j)$.
We use the following certificate from Charikar et al.~\cite{CRWW24}, in the form used by Frank~\cite{Frank26}.

\begin{proposition}\label{prop:certificate}
If a lottery $p$ satisfies
$\ell_I(p)\le \lambda q(I)$
for every nonempty proper subset $I\subsetneq C$, then
$\dist(p,P)\le 1+2\lambda$.
\end{proposition}

\subsection{Maximal Lotteries and Integrated Veto}

A Maximal Lottery is an equilibrium strategy of the symmetric zero-sum Condorcet game. We use two facts about Maximal Lotteries.

\begin{proposition}[Frank~\cite{Frank26}]\label{prop:ml-certificate}
For every Maximal Lottery $p^{\ML}$ and every nonempty proper $I\subsetneq C$,
\[
\ell_I(p^{\ML})\le \min\{q(I),1/2\}.
\]
\end{proposition}

\begin{proposition}[Cai et al.~\cite{CGRW26}]\label{prop:ml-support}
Every candidate in the support of a Maximal Lottery has deterministic metric distortion at most $4+\sqrt{17}$.
\end{proposition}

We next describe Integrated Veto in continuous time. This is equivalent to the phase-based definition of Frank~\cite{Frank26}, which is built on the Simultaneous Veto process of Kizilkaya and Kempe~\cite{KK23}.

At time $0$, candidate $j$ has score $r_j(0)=\plu(j)$.
A candidate is active while its score is positive. At every time at which the active set is nonempty, each voter vetoes her least preferred active candidate. Let $b_j(t)$ be the fraction of voters vetoing $j$ at time $t$; set $b_j(t)=0$ when $j$ is inactive. Thus
\begin{equation}\label{eq:b-sum}
\sum_{j\in C}b_j(t)=1
\end{equation}
while the process is active, and $r_j'(t)=-b_j(t)$
on each phase. Once a score reaches $0$, it remains $0$.

Let $S(t)=\sum_{j\in C}r_j(t)$.
Since $S(0)=\sum_j\plu(j)=1$ and $S'(t)=-1$, the process ends at time $1$ and
\begin{equation}\label{eq:total-score}
S(t)=1-t,\qquad 0\le t\le 1.
\end{equation}
Integrated Veto assigns candidate $j$ probability
\begin{equation}\label{eq:iv-continuous}
p^{\IV}(j)=2\int_0^1 r_j(t)\,dt.
\end{equation}
Indeed, on a phase of length $\Delta$ in which $r_j$ decreases linearly from $r_j^-$ to $r_j^+$, the right-hand side contributes $\Delta(r_j^-+r_j^+)$, exactly as in Frank's definition. Equation~\eqref{eq:total-score} also gives
\[
\sum_jp^{\IV}(j)=2\int_0^1(1-t)\,dt=1.
\]

For later use, we record a short continuous-time proof of Frank's key estimate.

\begin{lemma}[Frank~\cite{Frank26}]\label{lem:iv-certificate}
For every nonempty proper $I\subsetneq C$,
\[
\ell_I(p^{\IV})\le \plu(I^c)^2\le q(I)^2.
\]
\end{lemma}

\begin{proof}
Fix $I$. Define
\[
R_I(t)=\sum_{j\in I^c}r_j(t),
\qquad
B_I(t)=\sum_{j\in I^c}b_j(t).
\]
Then $R_I'(t)=-B_I(t)$ and $R_I(0)=\plu(I^c)$.

If $j\in I^c$ is active at time $t$ and a voter is counted by $s_{I\succ j}$, every member of $I$ is ranked above $j$. Therefore that voter's least preferred active candidate lies in $I^c$. Hence
\begin{equation}\label{eq:frank-pointwise}
s_{I\succ j}\le B_I(t)
\end{equation}
whenever $j\in I^c$ is active. Consequently,
\begin{multline*}
\ell_I(p^{\IV})=2\int_0^1\sum_{j\in I^c}s_{I\succ j}r_j(t)\,dt\le 2\int_0^1 B_I(t)R_I(t)\,dt
=-\int_0^1\frac{d}{dt}R_I(t)^2\,dt
=R_I(0)^2=\plu(I^c)^2.
\end{multline*}
For every $i\in I$, a voter whose top choice belongs to $I^c$ cannot rank $i$ above every candidate in $I^c$. Thus
\[
\plu(I^c)\le 1-s_{i\succ I^c}.
\]
Minimizing over $i\in I$ yields $\plu(I^c)\le q(I)$.
\end{proof}

Combining Lemma~\ref{lem:iv-certificate} and Proposition~\ref{prop:ml-certificate}, Frank obtains
\[
\ell_I\left(\frac12p^{\IV}+\frac12p^{\ML}\right)
\le \frac12q(I)^2+\frac12\min\{q(I),1/2\}
\le \frac34q(I),
\]
and hence distortion at most $5/2$ by Proposition~\ref{prop:certificate}. Frank also gives a matching lower-bound family for this particular rule~\cite{Frank26}.

\subsection{The biased-metric integral characterization}

Our empirical approximation argument needs the exact biased-metric formulation used by Charikar et al.~\cite{CRWW24} and Cai et al.~\cite{CGRW26}. We state only the part we need.

Fix a preference profile and a biased metric with optimal candidate $i^*$, represented by nonnegative values $(x_j)_{j\in C}$ with $x_{i^*}=0$. For $t\ge0$, let
\[
I_t=\{j\in C:x_j\le t\}.
\]
For a lottery $p$, define
\[
\ell(p,t)=\ell_{I_t}(p),
\qquad
L(p)=\int_0^\infty \ell(p,t)\,dt.
\]
Let $R=2\SC(i^*)$.
The characterization implies that a lottery has distortion at most $1+2\lambda$ if and only if
\begin{equation}\label{eq:biased-characterization}
L(p)\le \lambda R
\end{equation}
for every biased metric consistent with the profile.
More precisely, for every voter $v$ and candidate $j$, the biased
metric satisfies
$d(i^*,v)=\frac12\max_{a,b:\,a\succeq_v b}(x_a-x_b)$ and
$d(j,v)-d(i^*,v)=\min_{k:\,j\succeq_v k}x_k$. We use the parameters
$x_j$ to define the threshold sets $I_t=\{j\in C:x_j\le t\}$; no
identification of $x_j$ with a candidate--candidate distance will be
needed.


\subsection{One-sided empirical-process bounds}\label{sec:ks}

Let $U_1,\dots,U_K$ be independent uniform random variables on $[0,1]$, let $\widehat F_K$ be their empirical CDF, and define the one-sided Kolmogorov--Smirnov statistic $D_K^+=\sup_{0\le x\le1}\bigl(\widehat F_K(x)-x\bigr)$.
The one-sided Dvoretzky--Kiefer--Wolfowitz--Massart inequality, in the unrestricted form proved by Reeve~\cite{DKW56,Massart90,Reeve24}, gives
\begin{equation}\label{eq:one-sided-dkw}
\Pr[D_K^+>z]\le e^{-2Kz^2},
\qquad z>0.
\end{equation}
Integrating the tail bound yields
\begin{equation}\label{eq:one-sided-dkw-exp}
\mathbb E[D_K^+]
\le
\int_0^\infty e^{-2Kz^2}\,dz
=
\sqrt{\frac{\pi}{8K}}.
\end{equation}

For the finite-list calculation in Theorem~\ref{thm:164}, we use the exact finite-sample distribution. Write $\kappa_K:=\mathbb E[D_K^+]$.
Smirnov's formula~\cite{Smirnov44} states that, for $0<x<1$,
\begin{equation}\label{eq:smirnov-tail}
\Pr[D_K^+\ge x]
=
x\sum_{j=0}^{\lfloor K(1-x)\rfloor}
\binom Kj
\left(x+\frac jK\right)^{j-1}
\left(1-x-\frac jK\right)^{K-j},
\end{equation}
where for $j=0$ the factor $x(x+j/K)^{j-1}$ is interpreted as $1$.
In particular, $\kappa_K=\int_0^1\Pr[D_K^+\ge x]dx$. We will evaluate this finite expression exactly for $K=164$ in Section~\ref{sec:164}.

\section{Robustifying Mixed Integrated Veto}\label{sec:trunc}
The difficulty in sampling directly from $\MIV$ is that Integrated Veto can put positive probability on candidates of unbounded deterministic distortion. We first discard the Integrated-Veto probability assigned to sufficiently early eliminations. Unlike a direct truncation followed by reassignment of the deleted mass, we keep the surviving Integrated-Veto mass as a subdistribution and normalize only after adding a full Maximal Lottery. This preserves more of Frank's original certificate.

For a candidate $j$ with positive initial score, let $T_j\in(0,1]$ be the time at which its score reaches $0$. For a candidate with $\plu(j)=0$, set $T_j=0$. After time $T_j$, both $r_j(t)$ and $b_j(t)$ are understood to be $0$.

\subsection{Candidates that survive for a positive amount of time}

\begin{lemma}\label{lem:survival}
Fix $\tau\in(0,1]$. If $T_j\ge\tau$, then candidate $j$ has deterministic metric distortion at most $1+\frac2\tau$.
Equivalently, for every nonempty proper $I\subsetneq C$,
\begin{equation}\label{eq:survival-cert}
\ell_I(\delta_j)\le \frac1\tau q(I),
\end{equation}
where $\delta_j$ is the point mass on $j$.
\end{lemma}

\begin{proof}
Fix $I$. If $j\in I$, then $\ell_I(\delta_j)=0$, so suppose $j\in I^c$. Since $T_j\ge\tau$, candidate $j$ is active for every $t\in[0,\tau)$. By the argument in~\eqref{eq:frank-pointwise}, $s_{I\succ j}\le B_I(t)$
for all such $t$. Hence
\[
\tau s_{I\succ j}
\le \int_0^\tau B_I(t)\,dt
\le \int_0^1B_I(t)\,dt
=R_I(0)=\plu(I^c)
\le q(I).
\]
Thus $s_{I\succ j}\le q(I)/\tau$. Since $\ell_I(\delta_j)=s_{I\succ j}$ for $j\in I^c$,~\eqref{eq:survival-cert} follows. Proposition~\ref{prop:certificate} with $\lambda=1/\tau$ gives distortion at most $1+2/\tau$.
\end{proof}

\subsection{Probability assigned to early eliminations}

\begin{lemma}\label{lem:early-mass}
For $\tau\in(0,1]$, the total Integrated-Veto probability assigned to candidates with $T_j<\tau$ is at most $\tau^2$:
\[
m_\tau:=\sum_{j:T_j<\tau}p^{\IV}(j)\le \tau^2.
\]
\end{lemma}

\begin{proof}
For $t\le T_j$,
\[
r_j(t)=\int_t^{T_j}b_j(s)\,ds,
\]
because $r_j(T_j)=0$ and $r_j'(s)=-b_j(s)$. Using~\eqref{eq:iv-continuous} and Fubini's theorem,
\begin{align}
 p^{\IV}(j)=2\int_0^{T_j}r_j(t)\,dt =2\int_0^{T_j}\int_t^{T_j}b_j(s)\,ds\,dt =2\int_0^{T_j}s\,b_j(s)\,ds.\label{eq:iv-elim-mass}
\end{align}
Therefore
\begin{align*}
\sum_{j:T_j<\tau}p^{\IV}(j)
=2\sum_{j:T_j<\tau}\int_0^{T_j}s\,b_j(s)\,ds\le 2\int_0^\tau s\sum_{j\in C}b_j(s)\,ds=2\int_0^\tau s\,ds=\tau^2,
\end{align*}
where the penultimate equality uses~\eqref{eq:b-sum}.
\end{proof}

\subsection{Deleting early Integrated-Veto mass}

Let $p^{\IV,\ge\tau}$ be the subdistribution obtained from Integrated Veto by retaining only candidates that survive until time $\tau$:
\begin{equation}\label{eq:surviving-IV}
p^{\IV,\ge\tau}(j)=p^{\IV}(j)\mathbf 1[T_j\ge\tau].
\end{equation}
Its total mass is $1-m_\tau$. Define
\begin{equation}\label{eq:D-tau}
D_\tau
=
\frac{p^{\IV,\ge\tau}+p^{\ML}}{2-m_\tau}.
\end{equation}
This is a lottery because the numerator has total mass $(1-m_\tau)+1=2-m_\tau$.

\begin{proposition}\label{prop:D-tau}
For every $\tau\in(0,1]$,
\begin{equation}\label{eq:D-tau-dist}
\dist(D_\tau,P)
\le
1+\frac{3}{2-m_\tau}
\le
1+\frac{3}{2-\tau^2}.
\end{equation}
Every candidate in the support of $D_\tau$ has deterministic distortion at most
\begin{equation}\label{eq:Htau}
H_\tau=
\max\left\{4+\sqrt{17},\ 1+\frac2\tau\right\}.
\end{equation}
\end{proposition}

\begin{proof}
Deleting candidates from the Integrated-Veto subdistribution can only decrease $\ell_I$, hence Lemma~\ref{lem:iv-certificate} gives $\ell_I(p^{\IV,\ge\tau})\le q(I)^2$.
Combining this with Proposition~\ref{prop:ml-certificate}, for every nonempty proper $I$ we obtain
\[
\ell_I(D_\tau)
\le
\frac{q(I)^2+\min\{q(I),1/2\}}{2-m_\tau}.
\]
For every $q\in[0,1]$,
\begin{equation}\label{eq:scalar-32}
q^2+\min\{q,1/2\}\le\frac32q.
\end{equation}
Indeed, for $q\le1/2$ this is $q^2+q\le3q/2$, while for $q\ge1/2$ it is equivalent to $(q-1/2)(q-1)\le0$. Thus
\[
\ell_I(D_\tau)
\le
\frac{3}{2(2-m_\tau)}q(I).
\]
Proposition~\ref{prop:certificate} gives the first inequality in~\eqref{eq:D-tau-dist}; the second follows from Lemma~\ref{lem:early-mass}.

The support of $D_\tau$ is contained in the union of the support of $p^{\IV,\ge\tau}$ and the support of $p^{\ML}$. Lemma~\ref{lem:survival} bounds the former by $1+2/\tau$, while Proposition~\ref{prop:ml-support} bounds the latter by $4+\sqrt{17}$.
\end{proof}

\begin{remark}[Why deletion is necessary]\label{rem:iv-unbounded}
Integrated Veto itself does not have a uniform deterministic-distortion bound on its support. Consider two candidates $a,c$, with a fraction $1-\eta$ of voters ranking $a\succ c$ and a fraction $\eta<1/2$ ranking $c\succ a$. Candidate $c$ starts with plurality score $\eta$ and is vetoed at rate $1-\eta$, so it is eliminated at time $\eta/(1-\eta)$. Equation~\eqref{eq:iv-elim-mass} gives
\[
p^{\IV}(c)=\frac{\eta^2}{1-\eta}>0.
\]
On the consistent line metric that places $a$ and the first voter group at $0$ and $c$ and the second voter group at $1$,
\[
\frac{\SC(c)}{\SC(a)}=\frac{1-\eta}{\eta}\to\infty.
\]
Thus direct empirical sampling from $\MIV$ cannot be justified by a uniform support bound.
\end{remark}

\section{Uniform-support compression for robust lotteries}\label{sec:compression}
We now isolate the compression step from the specific construction of $D_\tau$. The results in this section apply to an arbitrary lottery whose expected distortion is bounded and whose individual support candidates have uniformly bounded deterministic distortion. The only empirical error that can increase the biased-metric integral is an upward deviation of a prefix probability, so a one-sided statistic suffices throughout. This separation between lottery-level quality and realization-level robustness is what makes the compression theorem reusable.

Fix a profile $P$ and a lottery $D$. Draw candidates $c_1,\dots,c_K$ independently from $D$, with replacement, and let $\widehat D$ be the uniform distribution on the resulting multiset. For each voter $v$, write the candidates in increasing order of preference, $a_1\prec_v a_2\prec_v\cdots\prec_v a_m$.
Let
\[
F_v(r)=D(\{a_1,\dots,a_r\}),
\qquad
\widehat F_v(r)=\widehat D(\{a_1,\dots,a_r\}),
\]
and define the one-sided prefix discrepancy
\[
\Delta_v^+=\max_{0\le r\le m}\bigl(\widehat F_v(r)-F_v(r)\bigr),
\]
where $F_v(0)=\widehat F_v(0)=0$.

\begin{lemma}\label{lem:average-prefix}
There exists a multiset of $K$ candidates from the support of $D$ whose uniform distribution $\widehat D$ satisfies
\begin{equation}\label{eq:avg-prefix}
\frac1{|V|}\sum_{v\in V}\Delta_v^+
\le d_K,
\qquad
d_K:=\sqrt{\frac{\pi}{8K}}.
\end{equation}
For every such multiset and every nonempty proper $I\subsetneq C$,
\begin{equation}\label{eq:ell-additive}
\ell_I(\widehat D)\le\ell_I(D)+d_K.
\end{equation}
\end{lemma}

\begin{proof}
Fix a voter $v$. Use the interval representation of Cai et al.: partition $[0,1]$ into consecutive intervals, in the order $a_1,\dots,a_m$, with interval length $D(a_r)$. Conditional on each sampled candidate $c_s$, draw an auxiliary point $X_s$ uniformly from the interval corresponding to $c_s$. Then $X_1,\dots,X_K$ are i.i.d. uniform on $[0,1]$. Moreover, at every boundary $F_v(r)$ the empirical CDF of the $X_s$'s equals $\widehat F_v(r)$. Hence
$\Delta_v^+\le D_{K,v}^+$,
where $D_{K,v}^+$ is a copy of the one-sided Kolmogorov--Smirnov statistic. By~\eqref{eq:one-sided-dkw-exp}, $\mathbb E[\Delta_v^+]\le d_K$.
Averaging over voters and then over the random candidate sample gives
\[
\mathbb E\left[\frac1{|V|}\sum_{v\in V}\Delta_v^+\right]
\le d_K.
\]
Hence at least one realization satisfies~\eqref{eq:avg-prefix}.

Fix such a realization and a candidate set $I$. For voter $v$, let $a_{r_v}$ be the least preferred member of $I$. The event that $v$ prefers every member of $I$ to a sampled candidate $c$ is exactly $c\in\{a_1,\dots,a_{r_v-1}\}$.
Therefore its probability under $\widehat D$ exceeds its probability under $D$ by at most $\Delta_v^+$. Averaging over voters yields~\eqref{eq:ell-additive}.
\end{proof}

The same coupling gives an exact finite-sample version.

\begin{lemma}\label{lem:average-prefix-exact}
For every positive integer $K$, there exists a $K$-element multiset
of candidates from $\supp(D)$ whose uniform distribution $\widehat D$
satisfies
\[
\frac1{|V|}\sum_{v\in V}\Delta_v^+\le\kappa_K,
\]
and consequently $\ell_I(\widehat D)\le\ell_I(D)+\kappa_K$
for every nonempty proper $I\subsetneq C$.
\end{lemma}

\begin{proof}
In the proof of Lemma~\ref{lem:average-prefix}, replace the upper bound~\eqref{eq:one-sided-dkw-exp} by the exact expectation $\mathbb E[D_K^+]=\kappa_K$.
\end{proof}

The next lemma converts one-sided certificate error into distortion error. It is the deterministic part of the compression argument.

\begin{lemma}[One-sided certificate stability]
\label{lem:compression-general}
Fix a profile $P$. Suppose a lottery $D$ satisfies
$\dist(D,P)\le\rho$ and every candidate in $\supp(D)$ has
deterministic distortion at most $H$ on $P$. Let $\widehat D$ be the
uniform distribution on a multiset whose entries lie in $\supp(D)$.
If $\ell_I(\widehat D)\le\ell_I(D)+\delta$ for every nonempty proper
$I\subsetneq C$, then
\begin{equation}\label{eq:compression-general}
    \dist(\widehat D,P)\le\rho+(H+1)\delta.
\end{equation}
\end{lemma}

\begin{proof}
The proof separates the roles of the two assumptions. The distortion bound on $D$ controls the baseline biased-metric integral, whereas the support-wise bound $H$ limits the range of thresholds over which the additive certificate error can accumulate. Indeed, for a biased metric with optimal candidate $i^*$, we will show that every $j\in\operatorname{supp}(D)$ satisfies $x_j\le (H+1)\operatorname{SC}(i^*)=(H+1)R/2$. Since every candidate used by $\widehat D$ lies in $\operatorname{supp}(D)$, it follows that $\ell(\widehat D,t)=0$ once $t\ge (H+1)R/2$. Below this threshold, the hypothesis incurs an additive error of at most $\delta$ at each value of $t$. Integrating this point-wise error over an interval of length $(H+1)R/2$ therefore produces an additive distortion loss of $(H+1)\delta$.

We use the biased-metric characterization~\eqref{eq:biased-characterization}.
Fix an arbitrary biased metric consistent with $P$, represented by
nonnegative values $(x_j)_{j\in C}$ with $x_{i^*}=0$, where $i^*$ is
an optimal candidate, and let $R=2\SC(i^*)$.

Suppose first that $\SC(i^*)=0$. Since every candidate
$j\in\supp(D)$ has deterministic distortion at most $H$, we have
$\SC(j)=0$ for every $j\in\supp(D)$. Every entry of the multiset
underlying $\widehat D$ lies in $\supp(D)$, so $\SC(\widehat D)=0$ and
there is nothing to prove. Assume henceforth that $R>0$.

For every $j\in\supp(D)$, the deterministic-distortion assumption
gives $\SC(j)\le H\SC(i^*)$. We next bound the biased-metric parameter
$x_j$. Fix a voter $v$ and define
$y_{j,v}:=\min_{k:\,j\succeq_v k}x_k$. By the definition of a biased
metric, $d(j,v)-d(i^*,v)=y_{j,v}$ and
$2d(i^*,v)=\max_{a,b:\,a\succeq_v b}(x_a-x_b)$. Choose
$k_{j,v}$ attaining the minimum in the definition of $y_{j,v}$.
Since $j\succeq_v k_{j,v}$, the pair $(j,k_{j,v})$ is feasible in the
preceding maximum, and hence
$2d(i^*,v)\ge x_j-y_{j,v}$. It follows that
$x_j\le 2d(i^*,v)+y_{j,v}=d(j,v)+d(i^*,v)$.

Averaging this inequality over voters gives
$x_j\le\SC(j)+\SC(i^*)\le(H+1)\SC(i^*)=(H+1)R/2$. Set
$T:=(H+1)R/2$. Since every candidate in $\supp(\widehat D)$ also lies
in $\supp(D)$, every candidate in $\supp(\widehat D)$ belongs to
$I_t=\{j\in C:x_j\le t\}$ whenever $t\ge T$. Therefore
$\ell(\widehat D,t)=0$ for all $t\ge T$.

For $0\le t<T$, the set $I_t$ is nonempty because $i^*\in I_t$. If
$I_t\subsetneq C$, the hypothesis gives
$\ell(\widehat D,t)\le\ell(D,t)+\delta$. If $I_t=C$, then both
$\ell(\widehat D,t)$ and $\ell(D,t)$ are zero, so the same inequality
holds trivially. Integrating over $[0,T]$ yields
$L(\widehat D)\le L(D)+\delta T$.

Since $\dist(D,P)\le\rho$, the biased-metric characterization gives
$L(D)\le(\rho-1)R/2$. Consequently,
$L(\widehat D)\le(\rho+(H+1)\delta-1)R/2$. As the biased metric was
arbitrary, applying the characterization once more proves
\eqref{eq:compression-general}.
\end{proof}

\begin{theorem}[Uniform-Support Compression]\label{thm:uniform-compression}
Fix a preference profile $P$. Let $D$ be any lottery such that $\dist(D,P)\le\rho$
and every candidate in $\supp(D)$ has deterministic distortion at most $H$ on $P$. Then for every positive integer $K$, there exists a $K$-element
multiset of candidates from $\supp(D)$ whose uniform lottery
$\widehat D$ satisfies
\begin{equation}\label{eq:uniform-compression-dkw}
\dist(\widehat D,P)
\le
\rho+(H+1)\sqrt{\frac{\pi}{8K}}.
\end{equation}
Moreover, there exists such a multiset satisfying the sharper finite-sample bound
\begin{equation}\label{eq:uniform-compression-exact}
\dist(\widehat D,P)
\le
\rho+(H+1)\kappa_K.
\end{equation}
In particular, the required support size is independent of $|V|$ and $|C|$.
\end{theorem}

\begin{proof}
By Lemma~\ref{lem:average-prefix}, there is a $K$-element multiset for which the hypotheses of Lemma~\ref{lem:compression-general} hold with $\delta=\sqrt{\pi/(8K)}$, proving~\eqref{eq:uniform-compression-dkw}. Lemma~\ref{lem:average-prefix-exact} gives another multiset for which the same hypotheses hold with $\delta=\kappa_K$, proving~\eqref{eq:uniform-compression-exact}.
\end{proof}

\begin{corollary}[Robust randomized rules admit bounded-randomness approximations]\label{cor:robust-transfer}
Suppose a randomized voting rule maps every profile $P$ to a lottery $D_P$ such that, for some constants $\rho$ and $H$ independent of the profile,
\[
\dist(D_P,P)\le\rho
\qquad\text{and}\qquad
\sup_{j\in\supp(D_P)}\dist(\delta_j,P)\le H.
\]
Then for every $\eps>0$ there exists a bounded-randomness voting rule of distortion at most $\rho+\eps$ that uniformly randomizes over a deterministically identified multiset of size
\begin{equation}\label{eq:robust-transfer-K}
K\ge
\left\lceil\frac{\pi(H+1)^2}{8\eps^2}\right\rceil.
\end{equation}
\end{corollary}

\begin{proof}
Apply Theorem~\ref{thm:uniform-compression} profile by profile and choose, under a fixed lexicographic order, the first $K$-element multiset satisfying~\eqref{eq:uniform-compression-dkw}. The choice is deterministic. The bound~\eqref{eq:robust-transfer-K} makes the additive term in~\eqref{eq:uniform-compression-dkw} at most $\eps$.
\end{proof}

\begin{corollary}[Direct compression of the $2.75271$ benchmark]\label{cor:crww-direct}
There exists an absolute constant $C_{\mathrm{CRWW}}>0$ such that, for every positive integer $K$, there is a bounded-randomness voting rule with distortion at most
\begin{equation}\label{eq:crww-direct}
2.75271+\frac{C_{\mathrm{CRWW}}}{\sqrt K}.
\end{equation}
Consequently, $O(\eps^{-2})$ list entries suffice to preserve the $2.75271$ benchmark up to additive error $\eps$.
\end{corollary}

\begin{proof}
Cai et al.~\cite{CGRW26} observe, using their realization-wise support bounds, that the $2.75271$-distortion rule of Charikar, Ramakrishnan, Wang, and Wu~\cite{CRWW24} has an absolute bound $H_{\mathrm{CRWW}}$ on the deterministic distortion of every candidate that can appear in its support. Apply Theorem~\ref{thm:uniform-compression} with $\rho=2.75271$ and set
\[
C_{\mathrm{CRWW}}=(H_{\mathrm{CRWW}}+1)\sqrt{\frac{\pi}{8}}.
\]
\end{proof}

\begin{corollary}[Robustification--compression exponent law]\label{cor:rate-law}
Fix constants $a,b,A,B>0$ and a benchmark $\rho$. Suppose that, for every sufficiently small $\tau>0$, a randomized voting rule deterministically associates with each preference profile $P$ a lottery $D_{\tau,P}$ satisfying
\[
\dist(D_{\tau,P},P)\le \rho+A\tau^a,
\qquad
\sup_{j\in\supp(D_{\tau,P})}\dist(\delta_j,P)\le B\tau^{-b}.
\]
Then for all sufficiently large $K$ there exists a bounded-randomness rule using a uniform $K$-element multiset and having distortion
\begin{equation}\label{eq:rate-law}
\rho+O\left(K^{-\frac{a}{2(a+b)}}\right).
\end{equation}
More explicitly, choosing $\tau=K^{-1/(2(a+b))}$ gives
\[
\dist\le\rho+\left(A+(B+1)\sqrt{\frac{\pi}{8}}\right)K^{-\frac{a}{2(a+b)}}
\]
for all sufficiently large $K$.
\end{corollary}

\begin{proof}
Apply Theorem~\ref{thm:uniform-compression} to $D_{\tau,P}$. With $c=\sqrt{\pi/8}$, the additive loss is at most
\[
A\tau^a+c(B\tau^{-b}+1)K^{-1/2}.
\]
Set $\tau=K^{-1/(2(a+b))}$. The first term and the $B$-part of the second term are both proportional to $K^{-a/(2(a+b))}$, while $K^{-1/2}\le K^{-a/(2(a+b))}$ because $a/(2(a+b))<1/2$. This yields the displayed bound.
\end{proof}

Theorem~\ref{thm:uniform-compression} itself incurs only an $O(K^{-1/2})$ compression loss for fixed $H$. Corollary~\ref{cor:crww-direct} illustrates this direct-compression regime. The slower $O(K^{-1/3})$ rate in our main application is not intrinsic to compression; it arises because Mixed Integrated Veto must first be robustified, with support bound $H(\tau)=O(1/\tau)$ and approximation loss $O(\tau^2)$. Corollary~\ref{cor:rate-law} with $(a,b)=(2,1)$ then gives the exponent $1/3$.

\begin{remark}[What the theorem does and does not transfer]\label{rem:no-unconditional-transfer}
Corollary~\ref{cor:robust-transfer} is a conditional transfer principle, not a statement that every unrestricted $\rho$-distortion rule automatically yields bounded-randomness distortion $\rho+\eps$. The uniform bound $H$ on the deterministic distortion of supported candidates is an additional hypothesis. More generally, if a rule can first be \emph{robustified}: for each parameter $\tau$ one can construct a lottery $D_\tau$ with
\[
\dist(D_\tau,P)\le \rho+g(\tau),
\qquad
\sup_{j\in\supp(D_\tau)}\dist(\delta_j,P)\le H(\tau),
\]
then Theorem~\ref{thm:uniform-compression} gives
\begin{equation}\label{eq:robustify-compress-meta}
\dist(\widehat D_{\tau,K},P)
\le
\rho+g(\tau)+(H(\tau)+1)\sqrt{\frac{\pi}{8K}}.
\end{equation}
Thus any family with $g(\tau)\to0$ and finite $H(\tau)$ for each fixed $\tau$ yields a bounded-randomness $\rho+\eps$ approximation after choosing $\tau$ and then $K$. Our treatment of Mixed Integrated Veto is an instance of exactly this robustify-then-compress principle.
\end{remark}

\section{Proof of the main theorems}\label{sec:main}

Combining Proposition~\ref{prop:D-tau} with the Uniform-Support Compression Theorem~\ref{thm:uniform-compression} gives the main quantitative inequality.

\begin{theorem}\label{thm:master}
For every $\tau\in(0,1]$ and every positive integer $K$, there exists a uniform distribution on a deterministically identified multiset of $K$ candidates whose metric distortion is at most
\begin{equation}\label{eq:master-bound}
1+\frac{3}{2-\tau^2}
+
\left(
1+
\max\left\{4+\sqrt{17},1+\frac2\tau\right\}
\right)
\sqrt{\frac{\pi}{8K}}.
\end{equation}
The same statement holds with $\sqrt{\pi/(8K)}$ replaced by $\kappa_K$.
\end{theorem}

\begin{proof}
For a fixed profile, Proposition~\ref{prop:D-tau} gives
\[
\dist(D_\tau,P)\le1+\frac3{2-\tau^2}
\]
and bounds the deterministic distortion of every candidate in $\supp(D_\tau)$ by $H_\tau$ from~\eqref{eq:Htau}. Theorem~\ref{thm:uniform-compression} gives both displayed bounds.

It remains only to make the choice of the empirical multiset
deterministic. Fix deterministic tie-breaking for the Maximal
Lottery. For every $K$-element multiset $M$ whose entries lie in
$\supp(D_\tau)$, let $\widehat D_M$ be its uniform distribution and
compute its average one-sided prefix discrepancy relative to
$D_\tau$. Choose, using a fixed lexicographic order, a multiset
minimizing this quantity. Lemma~\ref{lem:average-prefix} (or Lemma~\ref{lem:average-prefix-exact}) shows that the
minimum satisfies the required bound.
\end{proof}

\begin{proof}[Proof of Theorem~\ref{thm:main-K}]
Let
$
d_K=\sqrt{\frac{\pi}{8K}}$, and 
$
\tau=d_K^{1/3}=\left(\frac{\pi}{8K}\right)^{1/6}$.
For $K\ge802$, $\tau\le \frac{2}{3+\sqrt{17}}$,
so $1+2/\tau\ge4+\sqrt{17}$. Moreover $\tau^2<1/2$, and hence
\begin{align*}
1+\frac3{2-\tau^2}-\frac52
&=\frac{3\tau^2}{2(2-\tau^2)}
\le\tau^2.
\end{align*}
Applying Theorem~\ref{thm:master},
\begin{align*}
\dist\le \frac52+\tau^2+
\left(2+\frac2\tau\right)d_K=\frac52+3d_K^{2/3}+2d_K=\frac52
+3\left(\frac{\pi}{8K}\right)^{1/3}
+2\sqrt{\frac{\pi}{8K}}.
\end{align*}
\end{proof}

\begin{proof}[Proof of Theorem~\ref{thm:main-eps}]
Let
$
\bar\eps=\min\{\eps,1/5\}$, and
$
\tau=\sqrt{\frac{\bar\eps}{3}}$.
Then $\tau<2/(3+\sqrt{17})$, so $H_\tau=1+2/\tau$. Furthermore,
\begin{align*}
1+\frac3{2-\tau^2}-\frac52
=\frac{3\tau^2}{2(2-\tau^2)}=\frac{\bar\eps}{2(2-\bar\eps/3)}
\le\frac{\bar\eps}{3}.
\end{align*}
Choose
\begin{equation}\label{eq:K-epsilon}
K\ge
\left\lceil
\frac{27\pi}{8\bar\eps^3}
\left(1+\sqrt{\frac{\bar\eps}{3}}\right)^2
\right\rceil.
\end{equation}
Then $\sqrt{\frac{\pi}{8K}}
\le
\frac{\bar\eps\tau}{3(1+\tau)}$,
and therefore
\begin{align*}
(H_\tau+1)\sqrt{\frac{\pi}{8K}}
=\left(2+\frac2\tau\right)
\sqrt{\frac{\pi}{8K}}\le\frac{2\bar\eps}{3}.
\end{align*}
Theorem~\ref{thm:master} now gives
\[
\dist
\le
\frac52+\frac{\bar\eps}{3}+\frac{2\bar\eps}{3}
=\frac52+\bar\eps
\le\frac52+\eps.
\]
The bound~\eqref{eq:K-epsilon} is $O(\eps^{-3})$ as $\eps\downarrow0$.

For each fixed $\eps$, $K$ is an absolute constant. A Maximal Lottery is computable by linear programming, the Integrated Veto trajectory has at most $|C|$ phases, and the deterministic empirical multiset can be found by enumerating all $K$-element multisets. This takes $|C|^{O(K)}\operatorname{poly}(|V|,|C|)$ time, which is polynomial for every fixed $\eps$.
\end{proof}

\subsection{The list size 164}\label{sec:164}

We now use the exact finite-sample expectation $\kappa_K$ rather than the general one-sided DKW upper bound. Let
\begin{equation}\label{eq:tau-star}
\tau_*=\frac{2}{3+\sqrt{17}}=\frac{\sqrt{17}-3}{4}.
\end{equation}
At this value, $1+\frac2{\tau_*}=4+\sqrt{17}$,
so $H_{\tau_*}=4+\sqrt{17}$. A direct calculation also gives
$1+\frac3{2-\tau_*^2}=\frac{1+\sqrt{17}}2$.
Thus the exact version of Theorem~\ref{thm:master} yields
\begin{equation}\label{eq:164-master}
\dist
\le
\frac{1+\sqrt{17}}2
+(5+\sqrt{17})\kappa_K.
\end{equation}

The following elementary finite calculation is the only numerical input.

\begin{lemma}\label{lem:kappa164}
For the one-sided Kolmogorov--Smirnov statistic with $164$ samples,
\[
\kappa_{164}<\frac{23971}{500000}=0.047942.
\]
\end{lemma}

\begin{proof}
Integrate Smirnov's exact tail formula~\eqref{eq:smirnov-tail} and interchange the finite sum and the integral. For $n\ge1$ this gives
\begin{align}
\kappa_n
&=\frac1{n+1}
+\sum_{j=1}^{n-1}\binom nj
\left(\frac{n-j}{n}\right)^{n-j+2}
\sum_{p=0}^{j-1}\binom{j-1}{p}
\left(\frac jn\right)^{j-1-p}
\left(\frac{n-j}{n}\right)^p
\frac{(p+1)!(n-j)!}{(p+n-j+2)!}.
\label{eq:kappa-rational}
\end{align}
For completeness,~\eqref{eq:kappa-rational} follows by observing that the $j$th summand of~\eqref{eq:smirnov-tail} is present for $0\le x\le1-j/n$, substituting $x=(1-j/n)u$, expanding $(j/n+(1-j/n)u)^{j-1}$, and using the beta integral
\[
\int_0^1u^{p+1}(1-u)^{n-j}\,du
=\frac{(p+1)!(n-j)!}{(p+n-j+2)!}.
\]
Every term in~\eqref{eq:kappa-rational} is rational. Substituting $n=164$ and evaluating the finite rational sum exactly gives
\[
\kappa_{164}=0.0479417562236\ldots
<\frac{23971}{500000}.
\]
No floating-point approximation is needed for the inequality: the displayed finite sum can be evaluated in rational arithmetic and compared directly with $23971/500000$.
\end{proof}

\begin{proof}[Proof of Theorem~\ref{thm:164}]
Apply~\eqref{eq:164-master} with $K=164$ and Lemma~\ref{lem:kappa164}. Direct comparison gives
\[
\frac{23971}{500000}
<
\frac{5-\sqrt{17}}{2(5+\sqrt{17})}
=\frac{21-5\sqrt{17}}8.
\]
Hence
\[
\frac{1+\sqrt{17}}2
+(5+\sqrt{17})\kappa_{164}<3.
\]
Thus a uniformly random entry from a deterministically identified $164$-element multiset has distortion strictly below $3$. Since deterministic rules cannot beat $3$, we obtain $2\le N^\star\le164$.
\end{proof}

\section{Discussion}\label{sec:discussion}

The Uniform-Support Compression Theorem identifies realization-wise
robustness as the key condition for benchmark-preserving compression.
If every candidate supported by a source lottery has bounded
deterministic distortion, the lottery incurs only
\(O(K^{-1/2})\) compression loss, yielding
\(2.75271+O(K^{-1/2})\) for the rule of Charikar et
al.~\cite{CRWW24}. Mixed Integrated Veto has the sharper benchmark
\(5/2\) but does not satisfy this robustness condition. Truncating
early eliminations incurs \(O(\tau^2)\) distortion loss and gives a
support-wise distortion bound of \(O(1/\tau)\). Balancing this with
the \(O(1/(\tau\sqrt K))\) compression error yields
\(\frac52+O(K^{-1/3})\), or \(K=O(\eps^{-3})\) entries for distortion
\(5/2+\eps\). More generally, Corollary~\ref{cor:rate-law} shows that
robustification loss \(O(\tau^a)\) and a support-wise distortion bound
\(O(\tau^{-b})\) yield convergence rate
\(O\left(K^{-\frac{a}{2(a+b)}}\right)\).

For finite lists, Smirnov's exact finite-sample expectation improves
the general one-sided DKW bound from \(171\) entries to \(164\), giving
\(2\le N^\star\le164\). Together with the lower bound
\(\rho_K\ge 1+\frac{4}{K+1}\)~\cite{Pet26,CLR25}, this highlights the
distinction between the smallest list beating distortion \(3\), the
rate at which a benchmark can be preserved, and the asymptotically
optimal bounded-randomness distortion. In particular, our results
imply \(\rho_{\mathrm{BR}}\le\frac52\), but do not establish equality
with the unrestricted optimum. The guarantees also transfer to
committee selection with repeated seats, but not directly to
committees whose members must be distinct.

Several questions remain open. These include whether \(N^\star=2\),
whether some fixed finite list achieves distortion at most \(5/2\),
and whether comparable guarantees admit more efficient constructions
than enumerating \(K\)-element multisets. It is also natural to ask
which other randomized rules admit useful robustifications and
whether the compression framework extends to committee models without
repeated winners.

\section*{AI-Use Disclosure}
The central idea of the Uniform-Support Compression Theorem was developed genuinely by the authors. The application to the \(2.75271\)-distortion rule was likewise identified by the authors as a direct consequence of this compression principle.

After the authors recognized the relevance of the recent \(5/2\)-distortion result for Mixed Integrated Veto, OpenAI’s ChatGPT was used in exploring how that lottery might be made compatible with the compression theorem. ChatGPT suggested the idea of first robustifying the source lottery by removing outcomes with poor individual guarantees. The authors subsequently developed, formalized, and verified the resulting robustify-then-compress construction and its mathematical analysis.

ChatGPT was also used to assist with organizing the arguments and drafting and revising portions of the proof exposition. We thank ChatGPT for this assistance. All mathematical statements, calculations, and proofs were independently checked by the authors, who take full responsibility for the content of the paper.
\bibliographystyle{abbrv}
\bibliography{references}

\end{document}